\documentclass[pra,aps,floatfix,amsmath,superscriptaddress,onecolumn]{revtex4}

\usepackage{amssymb}
\usepackage{graphicx}
\usepackage{graphics}
\usepackage{amsmath}
\usepackage{amsthm}
\usepackage{color}
\usepackage{dsfont}
\usepackage[colorlinks,citecolor=blue,linkcolor=blue]{hyperref}
\usepackage{enumerate}
\usepackage{mathtools}
\usepackage[justification=centering]{caption}

\newtheorem{theorem}{Theorem}

\newtheorem{lemma}[theorem]{Lemma}
\newtheorem{cor}[theorem]{Corollary}

\newcommand{\bea}{\begin{eqnarray}}
\newcommand{\eea}{\end{eqnarray}}
\newcommand{\be}{\begin{equation}}
\newcommand{\ee}{\end{equation}}
\newcommand{\ba}{\begin{equation}\begin{aligned}}
\newcommand{\ea}{\end{aligned}\end{equation}}

\newtheorem{definition}{Definition}
\theoremstyle{remark}
\newtheorem{remark}{Remark}

\newtheorem{example}{Example}

\def\be{\begin{equation}}
\def\ee{\end{equation}}

\newcommand{\mA}{\mathcal{A}}

\newcommand{\C}{{\mathbb{C}}}

\newcommand{\ket}[1]{|#1\rangle}
\newcommand{\braket}[2]{\langle{#1}|{#2}\rangle}

\def\>{\rangle}
\def\<{\langle}

\def\diag{ \mathrm{diag}}

\begin{document}
	
	\title{Constructions of mutually unbiased  entangled bases}

	\author{Fei Shi}
	\email{shifei@mail.ustc.edu.cn}
	\affiliation{School of Cyber Security,
		University of Science and Technology of China, Hefei, 230026, People's Republic of China.}
	
	\author{Xiande Zhang}
	\email{drzhangx@ustc.edu.cn}
	\affiliation{School of Mathematical Sciences,
		University of Science and Technology of China, Hefei, 230026, People's Republic of China.}
	
	\author{Lin Chen}
	\email{linchen@buaa.edu.cn}
	\affiliation{School of Mathematical Sciences, Beihang University, Beijing 100191, China}
	\affiliation{International Research Institute for Multidisciplinary Science, Beihang University, Beijing 100191, China}
	
	\begin{abstract}
We construct two mutually unbiased bases  by maximally entangled states (MUMEB$s$)  in $\C^{2}\otimes \C^{3}$. This is the first example of  MUMEB$s$ in $\C^{d}\otimes \C^{d'}$ when $d\nmid d'$, namely $d'$ is not divisible by $d$. We show that they cannot be extended to four MUBs in $\C^6$. We propose a recursive construction of mutually unbiased  bases formed by special entangled states with a fixed Schmidt number $k$ (MUSEB$k$s). It shows that
 $\min \{t_{1},t_{2}\}$ MUSEB$k_{1}k_{2}$s in $\C^{pd}\otimes \C^{qd'}$ can be constructed from $t_{1}$ MUSEB$k_{1}$s in $\C^{d}\otimes \C^{d'}$ and $t_{2}$ MUSEB$k_{2}$s in $\C^{p}\otimes \C^{q}$ for any $d,d',p,q$. Further, we show that  three MUMEB$s$ exist in $\C^{d}\otimes \C^{d'}$ for any $d,d'$ with $d\mid d'$, and two  MUMEB$s$ exist in $\C^{d}\otimes \C^{d'}$ for infinitely many $d,d'$ with $d\nmid d'$.

	\end{abstract}
	
	
	\maketitle
\vspace{-0.5cm}
~~~~~~~~~~\indent{\it Keywords}: mutually unbiased bases, maximally entangled states, special entangled states
\section{Introduction}\label{intro}
The notion of mutually unbiased bases (MUBs) became  an essential feature of  quantum
mechanics in 1960 in the works of  Schwinger \cite{Schwinger}. In his work, Schwinger  realized that no
information can be retrieved when a quantum state which is prepared in a basis
state is measured with respect to the basis mutually unbiased with the prepared one.
This observation has a striking application  in the well-known BB84 quantum key distribution (QKD) protocol \cite{Bennett}.

Estimating an unknown physical transformation of a state, which is described by a unitary operator, was studied in \cite{Acin}. It is a more challenging problem compared to state estimation because it requires optimal measurements of states both post  and prior to transformation. Such estimation has an immediate application in the alignment of reference frames using quantum spins \cite{Chiribella2004}, and in the field of quantum cryptography but with a lower security threshold \cite{Lucamarini,Chiribella2008}. To improve the security, a notion of ``mutually unbiased unitary-operator bases" (MUUBs) was put forward in \cite{Scott} for unitary operators acting on a $d$-dimensional Hilbert space as  a set of pairwise \emph{mutually unbiased} unitary operator bases. The concept MUUBs was shown to be equivalent to MUBs formed by maximally entangled states (MUMEBs) \cite{Shaari}. In the same paper, Shaari et al. \cite{Shaari} generalized the original MUUBs for unitary operators acting on \emph{subspaces} of a Hilbert space, and considered the distinguishability of unitaries selected from a set of MUUBs and  its use in a QKD setup.

In this paper, we consider another generalization of MUUBs. Motivated by the equivalence between MUUBs in $\C^{d}$ and MUBs formed by maximally entangled states (MUMEBs)  in $\C^{d}\otimes \C^{d}$, we study  mutually unbiased bases formed by special entangled states with a fixed Schmidt number $k$ (MUSEB$k$s) in a general bipartite state space $\C^{d}\otimes \C^{d'}$ $(1\leq k\leq d,d')$,
Let $M_{k}(d,d')$ denote the maximum size of MUSEB$k$s  in $\C^{d}\otimes \C^{d'}$, i.e., the maximum number of mutually unbiased bases formed by special entangled states with Schmidt number $k$. The central problem here is to determine  $M_{k}(d,d')$ in a given bipartite state space and provide bounds for this value.

There have been results in \cite{Tao,Zhang,Liu,Xu,Xu1,Cheng,Han,Scott,Shaari} in which MUMEBs were known only when $d\mid d'$, namely $d'$ is divisible by $d$. We give the first example of MUMEBs when $d\nmid d'$, and  propose a simple construction of MUSEB$k$s for general systems  via matrix spaces. Consequently, we improve all previous results on MUSEB$k$s (MUMEBs).  See Table~\ref{all}  for a summary and comparison with known results.

 The rest of this paper is organized as follows. In Sec.~\ref{Def},  we introduce some related notations and terminologies, and give a one-to-one correspondence between MUSEB$k$s in $\C^{d}\otimes \C^{d'}$ and MUSEB$k$s in $\mathcal{M}_{d\times d'}$. Sec.~\ref{sec:first} provides a construction of two MUMEBs in $\C^{2}\otimes \C^{3}$, which is the first MUMEBs when $d\nmid d'$. In Sec.~\ref{sec:appli}, we show that the two MUMEBs in $\C^{2}\otimes \C^{3}$ constructed in Sec.~\ref{sec:first}  cannot be extended to 4 MUBs in $\C^6$.  In Sec.~\ref{Mks}, we establish a construction of
 $\min \{t_{1},t_{2}\}$ MUSEB$k_{1}k_{2}$s in $\C^{pd}\otimes \C^{qd'}$ from  $t_{1}$ MUSEB$k_{1}$s in $\C^{d}\otimes \C^{d'}$ and $t_{2}$ MUSEB$k_{2}$s in $\C^{p}\otimes \C^{q}$, which yields several new lower bounds in Sec.~\ref{Mbs}. We give the main conclusion and some open problems in Sec.~\ref{conpen}.

\begin{widetext}
	\begin{table}[htbp]
		\caption{Results about MUSEB$k$s in $\C^{d}\otimes \C^{d'}$}($d,q$ have a factorization of prime power, $p_1^{a_1}\leq\cdots\leq p_s^{a_s}$,$(p_1')^{a_1'}\leq\cdots\leq (p_s')^{a_s'}$)\label{all}
		\centering
		\renewcommand\tabcolsep{20.0pt}
		\begin{tabular}{ccc}
			\hline\noalign{\smallskip}
			Conditions                                         &Bounds                             & References\\
			\noalign{\smallskip}\hline\noalign{\smallskip}
			$d\geq2$                              &$M_d(d,d)\leq d^2-1$                                 &\cite{Scott}\\
			$d=2,3,5,7,11$                                   &$M_d(d,d)=d^2-1$                           &\cite{Scott} \\
			$d=p_{1}^{a_{1}}\cdots p_{s}^{a_{s}}$, $p_1^{a_1}\geq 3$     &$M_d(d,d)\geq p_1^{a_1}-1$       &\cite{Liu}\\
			$d\geq 3$, prime power &$M_d(d,d)\geq 2(d-1)$           &\cite{Xu}\\
			$d=p_{1}^{a_{1}}\cdots p_{s}^{a_{s}}$,  odd   &$ M_d(d,d)\geq 2(p_1^{a_1}-1)$             &\cite{Cheng}\\
			$d$, odd prime power                      &$M_d(d,d)\geq (d^2-1)/2$                          &\cite{Xu1}\\
			$d=2^{t}p_{1}^{a_{1}}\cdots p_{s}^{a_{s}}$, $t=0$        &$M_d(d,d)\geq(p_1^{2a_1}-1)/2$          &This paper\\
			$d=2^{t}p_{1}^{a_{1}}\cdots p_{s}^{a_{s}}$, $t=1$       &$M_d(d,d)\geq3$                                  &This paper\\
			$d=2^{t}p_{1}^{a_{1}}\cdots p_{s}^{a_{s}}$, $t\geq 2$ &$M_d(d,d)\geq \min\{2(2^t-1),(p_{1}^{2a_{1}}-1)/2\}$ &This paper\\\hline
			$M_d(d,d')\geq 2, \ell\geq 1$                  &$M_d(d,2^ld')\geq 2$                 &\cite{Zhang}\\
			$q$ and $ d$, both prime powers    &$M_{d}(d,qd)\geq \min\{q,M_{d}(d,d)\}$                    &\cite{Xu}\\
			$q=(p_{1}')^{a_{1}'}\cdots (p_{s}')^{a_{t}'}$, $d$ odd  &$M_{d}(d,qd)\geq         \min\{(p_{1}')^{a_{1}'}+1,M_{d}(d,d)\}$                                                         &\cite{Cheng} \\
			$d, q\geq 2$                  &$M_{d}(d,qd)\geq \min\{N(q),M_{d}(d,d)\}\geq 3$                 &This paper\\
			$d,q\geq 1$	                       &$M_{2d}(2d,3qd)\geq2$.                                        &This paper\\\hline
			$k=2^l$                              &$M_2(3,4k)\geq 2$                                              &\cite{Han}\\
			$k\geq 1$                             &$M_2(3,2k)\geq 2$                                             &This paper \\
			$d, p\geq 2$, $q\geq 1$                &$M_{d}(pd,qd)\geq \min\{N(p),M_{d}(d,qd)\}\geq 3$      &This paper   \\
			\noalign{\smallskip}\hline
		\end{tabular}\label{Res}
	\end{table}
\end{widetext}


\section{Definition and preliminary}\label{Def}


Assume that $1\leq k \leq \min\{d,d'\}$. A (pure) state $\ket{\psi}\in \C^{d}\otimes \C^{d'}$  is  called
a \emph{special entangled state with Schmidt number $k$} if it has  Schmidt decomposition:
\begin{equation*}
\psi=\frac{1}{\sqrt{k}}\sum_{i=0}^{k-1}\ket{i}\ket{i'},
\end{equation*}
where $\{\ket{i}\}$ and $\{\ket{i'}\}$ are orthonormal sets of $\C^{d}$ and $\C^{d'}$, respectively \cite{Shi}.
In particular, a special entangled state with Schmidt number $k$ is a product state when $k=1$, and a maximally entangled state when $k=\min\{d,d'\}$. A basis $\mathcal{B}$ is called a special entangled basis with  Schmidt number $k$ (SEB$k$) in $\C^{d}\otimes \C^{d'}$ if it contains $dd'$ pairwise orthogonal special entangled states with Schmidt number $k$ \cite{Shi,Guo15}. In particular, $\mathcal{B}$ is a maximally entangled state basis (MEB) when $k=\min\{d,d'\}$.

Two orthonormal bases $\mathcal{B}_{1}=\{\ket{\phi_{i}}\}_{i=1}^{dd'}$ and $\mathcal{B}_{2}=\{\ket{\psi_{i}}\}_{i=1}^{dd'}$  of $\C^{d}\otimes \C^{d'}$ are said to be \emph{mutually unbiased} if
\begin{equation*}
|\braket{\phi_{i}}{\psi_{j}}|=\frac{1}{\sqrt{dd'}}, \ (1\leq i, j\leq dd').
\end{equation*}
A set of SEB$k$s $\mathcal{A}=\{\mathcal{B}_{1}, \mathcal{B}_{2}, \ldots, \mathcal{B}_{m}\}$ in $\C^{d}\otimes \C^{d'}$ is called mutually unbaised (MUSEB$k$s) if every pair  $\mathcal{B}_{i}$ and $\mathcal{B}_{j}$ are mutually unbiased.  MUSEB$k$s are  MUUBs when $k=d=d'$, and MUMEBs when $k=\min\{d,d'\}$.
 When $k=d'=1$, MUSEB$k$s are just MUBs \cite{Andreas}, for which the maximum size is denoted by $N(d)$ conventionally.  It is shown that $N(d)\leq d+1$ and $N(d)=d+1$ when $d$ is a prime power \cite{Andreas}.  An open problem is to determine $N(d)$ when $d$ is not a prime power, which has an elementary lower bound ($d=p_1^{a_1}\dots p_s^{a_s}$, $p_1^{a_1}\leq\cdots\leq p_s^{a_s}$):
\begin{equation}\label{nd}
N(d)\geq \min\{N(p_{1}^{a_{1}}),\ldots ,N(p_{s}^{a_{s}})\}=p_{1}^{a_{1}}+1.
\end{equation}
Actually, this bound can be improved in some special cases \cite{Pawe}. It is shown that a complete set of MUBs of a bipartite system contains a fixed amount of entanglement \cite{wpz11}.
%

Now we define similar concepts in matrix spaces. Let $\mathcal{M}_{d\times d'}$ be the Hilbert  space of all $d\times d'$ complex matrices equipped with inner product
\begin{equation*}\label{inn}
(A,B)=\text{Tr}(A^{\dag}B).
\end{equation*}
There is a one-to-one relation between $\C^{d}\otimes \C^{d'}$ and $ \mathcal{M}_{d\times d'}$ \cite{Guo15}:
\begin{equation*}
\begin{aligned}
|\psi_{i}\rangle=&\sum_{k,l}a_{k,l}^{(i)}|k\rangle|l'\rangle\in \C^{d}\otimes \C^{d'} \Longleftrightarrow A_{i}=[a_{k,l}^{(i)}]\in \mathcal{M}_{d\times d'},\\
&Sr(|\psi_{i}\rangle)=\text{rank}(A_{i}), \ \langle\psi_{i}|\psi_{j}\rangle=\text{Tr}(A_{i}^{\dag}A_{j}),
\end{aligned}
\end{equation*}
where $\{k\}$ and $\{l'\}$ are the computational bases of $\C^{d}$ and $\C^{d'}$, respectively, and $Sr(|\psi_{i}\rangle)$ denotes the Schmidt number of $|\psi_{i}\rangle$.

 Let $\{\frac{1}{\sqrt{k}},\frac{1}{\sqrt{k}},\ldots,\frac{1}{\sqrt{k}}\}_{k}$ denote $k$ values of $\frac{1}{\sqrt{k}}$.
A $d\times d'$ matrix is called a \emph{$k$-singular-value-$\frac{1}{\sqrt{k}}$ matrix} if its nonzero singular values are $\{\frac{1}{\sqrt{k}},\frac{1}{\sqrt{k}},\ldots,\frac{1}{\sqrt{k}}\}_{k}$. Then $|\psi_{i}\rangle$ is a special entangled state with Schmidt number $k$ if and only if $A_{i}$ is a $k$-singular-value-$\frac{1}{\sqrt{k}}$ matrix. Especially, when $k=d=d'$, $|\psi_{i}\rangle$ is a maximally entangled state if and only if $\sqrt{d}A_{i}$ is a unitary matrix. We give the definitions of SEB$k$s and MUSEB$k$s in $\mathcal{M}_{d\times d'}$.

	\begin{definition}
	A set of $k$-singular-value-$\frac{1}{\sqrt{k}}$ matrices $\{A_{i}\}_{i=1}^{dd'}$ of $\mathcal{M}_{d\times d'}$ is called an SEB$k$ if $\text{Tr}(A_{i}^{\dag}A_{j})=\delta_{ij}$.
    \end{definition}
	
	\begin{definition}\label{defmu}
	Two SEB$k$s $\mathcal{B}_{1}=\{A_{i}\}_{i=1}^{dd'}$ and $\mathcal{B}_{2}=\{B_{i}\}_{i=1}^{dd'}$  in $\mathcal{M}_{d\times d'}$ are said to be \emph{mutually unbiased} if
	\begin{equation*}
	|\text{Tr}(A_{i}^{\dag}B_{j})|=\frac{1}{\sqrt{dd'}}, \ (1\leq i, j\leq dd').
	\end{equation*}
	A set of SEB$k$s $\mathcal{A}=\{\mathcal{B}_{1}, \mathcal{B}_{2}, \ldots, \mathcal{B}_{m}\}$ in $\mathcal{M}_{d\times d'}$ is called  MUSEB$k$s if every pair $\mathcal{B}_{i}$ and $\mathcal{B}_{j}$ are mutually unbiased.
    \end{definition}

Due to the one-to-one relation, an SEB$k$ of $\mathcal{M}_{d\times d'}$ is an SEB$k$ of  $\C^{d}\otimes \C^{d'}$;
MUSEB$k$s in $\mathcal{M}_{d\times d'}$ are MUSEB$k$s in $\C^{d}\otimes \C^{d'}$. When $k=d=d'$, Definition~\ref{defmu} gives  MUMEBs in $\mathcal{M}_{d\times d}$, which corresponds to  MUMEBs in $\C^{d}\otimes \C^{d}$. In this case, multiplying each matrix by a scaler $\sqrt{d}$ gives a set of mutually unbiased unitary bases (MUUBs) in $\C^{d}$ defined in \cite{Scott}. So Definition~\ref{defmu} can be viewed as a generation of MUUBs in  \cite{Scott}, from which we know $M_{d}(d,d)\leq d^{2}-1$ and $M_{d}(d,d)=d^{2}-1$ when $d=2,3,5,7,11$.	

\begin{remark}
  In \cite{Shaari}, the authors gave another definition of MUUBs:   Consider two distinct orthogonal bases, $\mA_0$ and $\mA_1$,
 composed of unitary transformations for some subspace
of the vector space $\C^{d}$. $\mA_0$ and $\mA_1$ are MUUBs
provided that \[|\text{Tr}(A_0^{(i)\dag}A_1^{(j)})|^2=c, ~\forall~A_0^{(i)}\in \mA_0, ~A_1^{(j)}\in \mA_1,\]
for $i,j = 1, . . . ,n$ and some constant $c\neq 0$.

When $\mA_0$ and $\mA_1$ are composed of unitary transformations for  the whole vector space $\C^{d}$, and $c=1$, then this definition reduces to Definition~\ref{defmu} for the  case $k=d=d'$ by  multiplying each matrix by a scaler $1/\sqrt{d}$.
\end{remark}

 \section{The first MUMEBs with $d\nmid d'$}\label{sec:first}
	
In previous works of \cite{Tao,Zhang,Liu,Xu,Xu1,Cheng}, all  constructions about MUMEBs are for $\C^{d}\otimes \C^{d'}$ with $d\mid d'$. Are there  MUMEBs  in $\C^{d}\otimes \C^{d'}$ when $d\nmid d'$?  We next give an example in this case.

\begin{theorem}
\label{C23}
	There are two MUMEBs in $\C^{2}\otimes \C^{3}$. Namely, $M_{2}(2,3)\geq 2$.
\end{theorem}
\begin{proof}
	Let
	$$\ket{\phi_{m,n}}=\frac{1}{\sqrt{2}}\sum_{p=0}^{1}(-1)^{np}\ket{p}\ket{(p\oplus m)'},$$ where $ n=0,1$, $m=0,1,2,$ and
	 $p\oplus m:=p+m \pmod 3$. It is easy to check that $\{\ket{\phi_{m,n}}\}$ is an MEB in $\C^{2}\otimes \C^{3}$.
	Let $A$ be a $2\times2$ unitary matrix
	$$
	A=\left(
	\begin{matrix}
	A_{00} & A_{01}\\
	A_{10} & A_{11}
	\end{matrix}\right),
	$$
	and let
	$$
	\left(
	\begin{matrix}
	\ket{a_{0}}\\
	\ket{a_{1}}
	\end{matrix}\right)=\left(
	\begin{matrix}
	A_{00} & A_{01}\\
	A_{10} & A_{11}
	\end{matrix}\right)\left(
	\begin{matrix}
	\ket{0}\\
	\ket{1}
	\end{matrix}\right).
	$$
	Denote
	$$\ket{\psi_{m,n}}=\frac{1}{\sqrt{2}}\sum_{p=0}^{1}(-1)^{np}\ket{a_{p}}\ket{(p\oplus m)'},$$ where $ n=0,1$ and $m=0,1,2.$
	Then $\{\ket{\psi_{m,n}}\}$ is also an MEB in $\C^{2}\otimes \C^{3}$.
	By definition, $\{\ket{\phi_{m,n}}\}$ and $\{\ket{\psi_{m,n}}\}$ are two MUMEBs  if and only if
	$$\left|\sum_{p_{2}=0}^{1}\sum_{p_{1}=0}^{1}(-1)^{\lambda}\braket{(p_{1}\oplus m_{1})'}{(p_{2}\oplus m_{2})'}A_{p_{2},p_{1}}\right|=\frac{2}{\sqrt{6}},$$for all $ \lambda=0,1$  and  $m_{1},m_{2}=0,1,2.$
	When $m_{1}=m_{2}$, we have  $|A_{0,0}\pm A_{1,1}|=\frac{2}{\sqrt{6}}$, which implies  $A_{0,0}\bot A_{1,1}$ and $|A_{0,0}|=|A_{1,1}|=\frac{1}{\sqrt{3}}$; when $m_{2}\equiv m_{1}+1 \pmod 3$,  we have  $|A_{0,1}|=\frac{2}{\sqrt{6}}$; when $m_{1}\equiv m_{2}+1 \pmod 3$,  we have  $|A_{1,0}|=\frac{2}{\sqrt{6}}$.
	So we can assume that
	\begin{equation}\label{aia}
	A=\frac{1}{\sqrt{3}}\left(
	\begin{matrix}
	e^{i\theta_{1}} &\sqrt{2}e^{i\theta_{2}}\\
	\sqrt{2}e^{i\theta_{3}} &e^{i(\theta_{1}+\frac{\pi}{2})}
	\end{matrix}\right),
	\end{equation} for some $\theta_i$, $i=1,2,3$.
	Since $A$ is a unitary matrix, $\theta_{i}$ must satisfy the following condition:
	\begin{equation}\label{iai}
	\theta_{2}+\theta_{3}-2\theta_{1}=2k\pi+\frac{3\pi}{2},   \text{ for some }  k\in Z.
	\end{equation}
	Taking
	$$\theta_{1}=0,\ \theta_{2}=\frac{3\pi}{2},\ \theta_{3}=0,$$
	we have
	$$
	A=\frac{1}{\sqrt{3}}\left(
	\begin{matrix}
	1        &-\sqrt{2}i\\
	\sqrt{2} &i
	\end{matrix}\right) \text{ and}
	$$
	
	$$
	\ket{a_{0}}=\frac{1}{\sqrt{3}}(\ket{0}-\sqrt{2}i\ket{1}),\ \ket{a_{1}}=\frac{1}{\sqrt{3}}(\sqrt{2}\ket{0}+i\ket{1}).
	$$
	Then $\{\ket{\phi_{m,n}}\}$ and $\{\ket{\psi_{m,n}}\}$ are as follows:
	\begin{equation}\label{16}
	\left\{
	\begin{aligned}
	&\ket{\phi_{0,0}}=\frac{1}{\sqrt{2}}(\ket{00'}+\ket{11'})\\
	&\ket{\phi_{0,1}}=\frac{1}{\sqrt{2}}(\ket{00'}-\ket{11'})\\
	&\ket{\phi_{1,0}}=\frac{1}{\sqrt{2}}(\ket{01'}+\ket{12'})\\
	&\ket{\phi_{1,1}}=\frac{1}{\sqrt{2}}(\ket{01'}-\ket{12'})\\
	&\ket{\phi_{2,0}}=\frac{1}{\sqrt{2}}(\ket{02'}+\ket{10'})\\
	&\ket{\phi_{2,1}}=\frac{1}{\sqrt{2}}(\ket{02'}-\ket{10'})
	\end{aligned}
	\right. ,
	\end{equation}
	\begin{equation}\label{17}
	\left\{
	\begin{aligned}
	&\ket{\psi_{0,0}}=\frac{1}{\sqrt{6}}(\ket{00'}+\sqrt{2}\ket{01'}-\sqrt{2}i\ket{10'}+i\ket{11'})\\
	&\ket{\psi_{0,1}}=\frac{1}{\sqrt{6}}(\ket{00'}-\sqrt{2}\ket{01'}-\sqrt{2}i\ket{10'}-i\ket{11'})\\
	&\ket{\psi_{1,0}}=\frac{1}{\sqrt{6}}(\ket{01'}+\sqrt{2}\ket{02'}-\sqrt{2}i\ket{11'}+i\ket{12'})\\
	&\ket{\psi_{1,1}}=\frac{1}{\sqrt{6}}(\ket{01'}-\sqrt{2}\ket{02'}-\sqrt{2}i\ket{11'}-i\ket{12'})\\
	&\ket{\psi_{2,0}}=\frac{1}{\sqrt{6}}(\ket{02'}+\sqrt{2}\ket{00'}-\sqrt{2}i\ket{12'}+i\ket{10'})\\
	&\ket{\psi_{2,1}}=\frac{1}{\sqrt{6}}(\ket{02'}-\sqrt{2}\ket{00'}-\sqrt{2}i\ket{12'}-i\ket{10'})
	\end{aligned}
	\right. .
	\end{equation}
	
	Thus the above $\{\ket{\phi_{m,n}}\}$ and $\{\ket{\psi_{m,n}}\}$ are two MUMEBs in $\C^{2}\otimes \C^{3}$.
\end{proof}

Note that 	we can also take
	$\theta_{1}=\frac{\pi}{4},\ \theta_{2}=\frac{3\pi}{4},\ \theta_{3}=\frac{5\pi}{4}$ in Eq. (\ref{aia}), then
	$$
	A=\frac{1}{\sqrt{3}}\left(
	\begin{matrix}
	\frac{\sqrt{2}}{2}(1+i)        &-1+i\\
	-1-i&\frac{\sqrt{2}}{2}(-1+i)
	\end{matrix}\right) \text{ and}
	$$
\begin{equation}\label{othermu}
  \ket{a_{0}}  =\frac{1}{\sqrt{3}}(\frac{\sqrt{2}}{2}(1+i)\ket{0}+(-1+i)\ket{1}),
  \ket{a_{1}} =\frac{1}{\sqrt{3}}((-1-i)\ket{0}+\frac{\sqrt{2}}{2}(-1+i)\ket{1}).
\end{equation}
	The above $\{\ket{\phi_{m,n}}\}$ and $\{\ket{\psi_{m,n}}\}$ are also two MUMEBs in $\C^{2}\otimes \C^{3}$. Unfortunately, if $A_{1}$ and $A_{2}$ are two matrices that satisfy the conditions of Eq.~(\ref{aia}) and Eq.~(\ref{iai}), then $A_{1}A_{2}$ cannot satisfy these conditions. It means we cannot construct three MUMEB$s$ in $\C^{2}\otimes \C^{3}$ by this way.

 We should mention that our method of constructing two MUMEBs in $\C^{2}\otimes \C^{3}$  in Theorem~\ref{C23} is similar to that  of constructing MUMEBs in $\C^{d}\otimes \C^{kd}$ from \cite{Tao}, where  a unitary matrix $A$ of size $kd\times kd$ is applied on the latter subspace $\C^{kd}$. However,  using their method, we find that such a $3\times 3$ unitary  matrix $A$ does not exist. So we cannot construct MUMEBs in $\C^{2}\otimes \C^{3}$  by the method in \cite{Tao}.


Before closing this section, we convert $\{\ket{\phi_{m,n}}\}$ and $\{\ket{\psi_{m,n}}\}$ of Eq.~(\ref{16}) and Eq.~(\ref{17}) to  MUMEBs $\mathcal{R}_{1}$ and $\mathcal{R}_{2}$ in $\mathcal{M}_{2\times 3}$,  respectively.
\begin{align*}
\mathcal{R}_{1}=\{&\frac{1}{\sqrt{2}}\left(
\begin{matrix}
1  & 0  & 0\\
0  & 1  & 0
\end{matrix}\right), \ \frac{1}{\sqrt{2}}\left(
\begin{matrix}
1  & 0 & 0\\
0  & -1 & 0
\end{matrix}\right), \ \frac{1}{\sqrt{2}}\left(
\begin{matrix}
0  & 1     & 0\\
0  & 0     & 1
\end{matrix}\right), \\
&\frac{1}{\sqrt{2}}\left(
\begin{matrix}
0   & 1  & 0\\
0   & 0  & -1
\end{matrix}\right), \ \frac{1}{\sqrt{2}}\left(
\begin{matrix}
0     & 0   & 1\\
1     & 0   & 0
\end{matrix}
\right), \ \frac{1}{\sqrt{2}}\left(
\begin{matrix}
0       & 0      & 1\\
-1       & 0      &0
\end{matrix}
\right)\},\\
\mathcal{R}_{2}=\{&\frac{1}{\sqrt{6}}\left(
\begin{matrix}
1          & \sqrt{2} & 0\\
-\sqrt{2}i & i       & 0
\end{matrix}\right), \ \frac{1}{\sqrt{6}}\left(
\begin{matrix}
1          & -\sqrt{2} & 0\\
-\sqrt{2}i & -i       & 0
\end{matrix}
\right), \\
&\frac{1}{\sqrt{6}}\left(
\begin{matrix}
0  & 1   & \sqrt{2}\\
0  & -\sqrt{2}i   & i
\end{matrix}
\right),
\frac{1}{\sqrt{6}}\left(
\begin{matrix}
0  & 1   & -\sqrt{2}\\
0  & -\sqrt{2}i   & -i
\end{matrix}\right), \\
&\frac{1}{\sqrt{6}}\left(
\begin{matrix}
\sqrt{2}& 0   & 1\\
i     & 0   &-\sqrt{2}i
\end{matrix}
\right), \ \frac{1}{\sqrt{6}}\left(
\begin{matrix}
-\sqrt{2}& 0   & 1\\
-i     & 0   &-\sqrt{2}i
\end{matrix}
\right)\}.
\end{align*}

\section{Application}\label{sec:appli}
In this section, we propose an application of the two MUMEBs in $\C^2\otimes\C^3$ constructed in Eqs. (\ref{16}) and (\ref{17}) in Theorem \ref{thm:mumeb}. This is to tackle the special case of a long-standing open problem, namely whether four six-dimensional MUBs exist. The study of MUBs has extensive physical applications in quantum cryptography, tomography, and the construction of Wigner functions. It has been shown that there exist three six-dimensional MUBs, and widely believed that four six-dimensional MUBs may not exist in spite of great efforts from mathematical and physical community
\cite{Boykin05,bw08,bw09,jmm09,mub09,bw10,deb10,wpz11,mw12ijqi,mw12jpa135307,rle11,Sz12,Goyeneche13,mw12jpa102001,mb15,mpw16,
	cy17,chen2018mutually,liang2019h_2}.

To explain the application, we refer to the vectors in an orthonormal basis in $\C^n$ as the column vectors of a unitary matrix. In particular we refer to the computational basis as the identity matrix. Hence, the existence of four six-dimensional MUBs is equivalent to the existence of the identity matrix, and three complex Hadamard matrices (CHMs) which are MU. If they exist then we refer to the three CHMs as the so-called MUB trio \cite{cy17}. From Lemma 11 (ii) and matrix $Y_6$ of \cite{cy17}, we have found the following criterion of determining when a CHM belongs to an MUB trio.
\begin{lemma}
	\label{le:mub=trio}
	\begin{enumerate}[(i)]
	\item  A CHM belongs to an MUB trio if and only if so does its transpose.
	
	\item  If a CHM has a $2\times3$ real submatrix then it does not belong to any MUB trio.
		\end{enumerate}
\end{lemma}
Now we are in a position to present the application of the two MUMEBs in Eqs. (\ref{16}) and (\ref{17}).
\begin{theorem}
	\label{thm:mumeb}	
	The two MUMEBs in Eqs. (\ref{16}) and (\ref{17}) cannot be extended to four MUBs.	
\end{theorem}
\begin{proof}
	In the following, we express the two MUMEBs as two unitary matrices $U$ and $V$, respectively.
	$$
	U=	\left(
	\begin{matrix}
	\frac{1}{\sqrt{2}} & \frac{1}{\sqrt{2}} & 0 & 0 & 0 & 0 \\
	0&0 & \frac{1}{\sqrt{2}} & \frac{1}{\sqrt{2}} &0 &0 \\
	0&0 &0 &0 & \frac{1}{\sqrt{2}} & \frac{1}{\sqrt{2}} \\
	0&0 &0 &0 & \frac{1}{\sqrt{2}} & -\frac{1}{\sqrt{2}}\\
	\frac{1}{\sqrt{2}} & -\frac{1}{\sqrt{2}} & 0&0 &0 &0 \\
	0&0 & \frac{1}{\sqrt{2}} & -\frac{1}{\sqrt{2}} & 0& 0	
	\end{matrix}
	\right),
	$$
	and
	$$
	V=\left(
	\begin{matrix}
	\frac{1}{\sqrt{6}} & \frac{1}{\sqrt{6}} & 0 & 0 & \frac{\sqrt{2}}{\sqrt{6}} & -\frac{\sqrt{2}}{\sqrt{6}} \\
	\frac{\sqrt{2}}{\sqrt{6}} & -\frac{\sqrt{2}}{\sqrt{6}} &
		\frac{1}{\sqrt{6}} & \frac{1}{\sqrt{6}} & 0 & 0 \\
		0 & 0 & \frac{\sqrt{2}}{\sqrt{6}}  & -\frac{\sqrt{2}}{\sqrt{6}} & \frac{1}{\sqrt{6}} & \frac{1}{\sqrt{6}} \\
		-\frac{\sqrt{2}i}{\sqrt{6}}  & -\frac{\sqrt{2}i}{\sqrt{6}}  & 0 & 0 & \frac{i}{\sqrt{6}} & -\frac{i}{\sqrt{6}} \\
		\frac{i}{\sqrt{6}} & -\frac{i}{\sqrt{6}} & -\frac{\sqrt{2}i}{\sqrt{6}} &-\frac{\sqrt{2}i}{\sqrt{6}} & 0 & 0 \\
		0& 0& \frac{i}{\sqrt{6}} & -\frac{i}{\sqrt{6}} &
		-\frac{\sqrt{2}i}{\sqrt{6}}  & -\frac{\sqrt{2}i}{\sqrt{6}}
		\end{matrix}
		\right).
	$$
	
	One can verify that the lower left $2\times3$ submatrix of $U^{\dagger} VQ$ is
	$\left(
	\begin{matrix}
	-\frac{1}{\sqrt{6}} & -\frac{1}{\sqrt{6}} & \frac{1}{\sqrt{6}}	\\
     \frac{1}{\sqrt{6}} &\frac{1}{\sqrt{6}}&\frac{1}{\sqrt{6}}
     \end{matrix}
	\right)$, where the diagonal unitary $Q=\diag(-i,-i,1,1,1,1)$. Hence $U^{\dagger} V Q$ does not belong to any MUB trio in terms of Lemma \ref{le:mub=trio}. We have proven the assertion.
\end{proof}
Our results restrict the form of MUB trio, if it really exists.  Following the proof of Theorem \ref{thm:mumeb}, one can similarly show that the other pair of MUMEBs constructed at Eq.~(\ref{othermu}) cannot be extended to four MUBs too. We may conjecture that, any two MUMEBs cannot be extended to four MUBs.

\section{A Recursive Construction of MUSEBks}\label{Mks}

In this section, we give a construction of MUSEB$k_{1}k_{2}$s in $\C^{pd}\otimes \C^{qd'}$ from  MUSEB$k_{1}$s in $\C^{d}\otimes \C^{d'}$ and MUSEB$k_{2}$s in $\C^{p}\otimes \C^{q}$ for any positive integers $d$, $d'$, $p$ and $q$.

For any two sets $\mathcal{S}=\{A_{1},A_{2},\ldots,A_{n}\}$ and  $\mathcal{T}=\{B_{1},B_{2},\ldots,B_{m}\}$, define $\mathcal{S}^{\mathrm{T}}:=\{A_{1}^{\mathrm{T}},A_{2}^{\mathrm{T}},\cdots,A_{n}^{\mathrm{T}}\},$ and
$$\mathcal{S}\otimes\mathcal{T}:=\{A_{i}\otimes B_{j}\mid i=1,\ldots, n, j=1,\ldots,m\}.$$

\begin{theorem}
\label{mum}
	If there are $t_{1}$ MUSEB$k_{1}$s in $\C^{d}\otimes \C^{d'}$ and $t_{2}$ MUSEB$k_{2}$s in $\C^{p}\otimes \C^{q}$, then there are $ \min \{t_{1},t_{2}\}$ MUSEB$k_{1}k_{2}$s in $\C^{pd}\otimes \C^{qd'}$. Namely,
	\begin{equation}
	M_{k_{1}k_{2}}(pd,qd')\geq \min\{M_{k_{1}}(d,d'),M_{k_{2}}(p,q)\}.
	\end{equation}
\end{theorem}
\begin{proof}We prove it by using the language of MUSEB$k$s in $\mathcal{M}_{d\times d'}$. Let $\mathcal{S}_{1}=\{A_{1},A_{2},\ldots,A_{dd'}\}$ and $\mathcal{S}_{2}=\{B_{1},B_{2},\ldots,B_{dd'}\}$ be two  MUSEB$k_{1}$s in $\mathcal{M}_{d\times d'}$, that is, $A_{i},B_{j}$ are $k_1$-singular-value-$\frac{1}{\sqrt{k_1}}$ matrices,  $Tr(A_{i}^{\dag}A_{j})=\delta_{i,j}$, $Tr(B_{i}^{\dag}B_{j})=\delta_{i,j}$ and $|\text{Tr}(A_{i}^{\dag}B_{j})|=\frac{1}{\sqrt{dd'}}$.
	
	Similarly, let $\mathcal{T}_{1}=\{C_{1},C_{2},\ldots,C_{pq}\}$ and  $\mathcal{T}_{2}=\{D_{1},D_{2},\ldots,D_{pq}\}$  be two  MUSEB$k_{2}$s in $\mathcal{M}_{p\times q}$, then $C_{i},D_{j}$ each have $k_{2}$  singular values $\frac{1}{\sqrt{k_{2}}}$, $Tr(C_{i}^{\dag}C_{j})=\delta_{i,j}$, $Tr(D_{i}^{\dag}D_{j})=\delta_{i,j}$ and  $|\text{Tr}(C_{i}^{\dag}D_{j})|=\frac{1}{\sqrt{pq}}$.
	
	Then $\mathcal{S}_{1}\otimes \mathcal{T}_{1}$ and $\mathcal{S}_{2}\otimes \mathcal{T}_{2}$ are two sets of $dd'pq$ matrices in $\mathcal{M}_{pd\times qd'}$. We prove that they are two  MUSEB$k_{1}k_{2}$s in $\mathcal{M}_{pd\times qd'}$.
	
	First, we show that $\mathcal{S}_{1}\otimes \mathcal{T}_{1}$ and $\mathcal{S}_{2}\otimes \mathcal{T}_{2}$ are two SEB$k_{1}k_{2}$s. We only prove it for $\mathcal{S}_{1}\otimes \mathcal{T}_{1}$, the other one is similar.  Since each $A_{i} \otimes C_{j}\in \mathcal{S}_{1}\otimes \mathcal{T}_{1}$ has nonzero singular values $\{\frac{1}{\sqrt{k_{1}k_{2}}},\frac{1}{\sqrt{k_{1}k_{2}}},\ldots,\frac{1}{\sqrt{k_{1}k_{2}}}\}_{k_{1}k_{2}}$, it is a $k_1k_2$-singular-value-$\frac{1}{\sqrt{k_1k_2}}$ matrix. Next, $\text{Tr}[(A_{i}\otimes C_{j})^{\dag}(A_{i'}\otimes C_{j'})]=\text{Tr}(A_{i}^{\dag}A_{i'})\cdot\text{Tr}(C_{j}^{\dag}C_{j'})=\delta_{i,i'}\delta_{j,j'}$, so they are orthogonal if $(i,j)\neq (i',j')$.
	
	Finally, the mutually unbiased property follows from the fact that $|\text{Tr}[(A_{i}\otimes C_{j})^{\dag}(B_{i'}\otimes D_{j'})]|=|\text{Tr}(A_{i}^{\dag}B_{i'})|\cdot|\text{Tr}(C_{j}^{\dag}D_{j'})|=\frac{1}{\sqrt{pdqd'}}$.

	Let $t=\min\{t_{1},t_{2}\}$. Assume that $\mathcal{S}_{1},\mathcal{S}_{2},\ldots,\mathcal{S}_{t}$ are $t$ MUSEB$k_{1}$s in $\mathcal{M}_{d\times d'}$, and $\mathcal{T}_{1},\mathcal{T}_{2},\ldots,\mathcal{T}_{t}$ are $t$ MUSEB$k_{2}$s in $\mathcal{M}_{p\times q}$.
	It is easy to see that $\{\mathcal{S}_{1}\otimes \mathcal{T}_{1}, \mathcal{S}_{2}\otimes \mathcal{T}_{2},\ldots \mathcal{S}_{t}\otimes \mathcal{T}_{t}\}$ is a set of $t$ MUSEB$k_{1}k_{2}$s in $\mathcal{M}_{pd\times qd'}$ from the above analysis.
\end{proof}

Theorem~\ref{mum} gives a useful generic  construction of MUSEB$k$s (MUMEBs).  We illustrate its importance in several cases.

	When $k_1=d\leq d'$ and $k_2=p\leq q$, we obtain the case for MUMEBs,
	\begin{equation}\label{meb}
	M_{pd}(pd,qd')\geq \min\{M_{d}(d,d'),M_{p}(p,q)\}.
	\end{equation}

\begin{example}\label{rem1}	There are five MUMEBs in $\C^{2}\otimes \C^{4}$ and three MUMEBs in $\C^{2}\otimes \C^{6}$ \cite{Tao}, then we can construct at least three MUMEBs in $\C^{4}\otimes \C^{24}$, that is, $M_{4}(4,24)\geq \min\{M_{2}(2,4),M_{2}(2,6)\}\geq 3$.
\end{example}

%
%

Observing that Theorem~\ref{mum} does not require $d\leq d'$ and $p\leq q$, and the fact that $M_{k}(d,d')=M_{k}(d',d)$, we can get the following different example.

\begin{example}
From five MUMEBs in $\C^{2}\otimes \C^{4}$ and three MUMEBs in $\C^{2}\otimes \C^{6}$, we can also construct three MUSEB$4$s in $\C^{8}\otimes \C^{12}$, that is, $M_{4}(8,12)\geq \min\{M_{2}(4,2),M_{2}(2,6)\}\geq 3$.
\end{example}

Now we consider Theorem~\ref{mum} when some parameters of $d,d',p,q$ equal one. Observing that $\mathcal{M}_{1\times d'}$  is exactly the same  space as $\C^{d'}$,  an SEB$1$ in $\mathcal{M}_{1\times d'}$ is an orthonormal basis in $\C^{d'}$,  and MUSEB$1$s in $\mathcal{M}_{1\times d'}$ are indeed MUBs in $\C^{d'}$. Namely,  $M_{1}(1,d')=N(d')$. Hence,  Eq~(\ref{nd}) is a special case of  Theorem~\ref{mum}  when $k_1=k_2=d=p=1$. When $p=1$ (or $q=1$), we have the following corollary.

\begin{cor}
\label{seb}
	If there are $t_{1}$ MUSEBks in $\C^{d}\otimes \C^{d'}$, $t_{2}$ MUBs in $\C^{q}$, then there are at least $\min\{t_{1},t_{2}\}$ MUSEBks in $\C^{d}\otimes \C^{qd'}$. Namely,
	\begin{equation}\label{k1k}
	M_{k}(d,qd')\geq \min\{N(q),M_{k}(d,d')\}.
	\end{equation}
	If $k=d\leq d'$, then
	\begin{equation}\label{mdq}
	M_{d}(d,qd')\geq \min\{N(q),M_{d}(d,d')\},
	\end{equation}
and
	\begin{equation}\label{mpd}
	M_{d}(pd,d')\geq \min\{N(p),M_{d}(d,d')\}.
	\end{equation}
\end{cor}

\begin{remark}
	As $N(q)\geq 3$ for general $q\geq 2$ \cite{Andreas}, by  Eq.~(\ref{mdq}), \emph{if $M_{d}(d,d')\geq 2$, then $M_{d}(d,qd')\geq 2$ for any $q\geq 2$.}
	This improves the result in \cite{Zhang}, which only showed that if $M_{d}(d,d')\geq 2$, then $M_{d}(d,2^{l}d')\geq 2$.
\end{remark}


	From Eq.~(\ref{mpd}), we give a method to construct MUSEB$k$s  from MUBs and MUMEBs. In fact, we can construct MUSEB$d$s in $\C^{pd}\otimes \C^{d'}$ from MUBs in $\C^{p}$ and MUMEBs in $\C^{d}\otimes \C^{d'}$ for any $p\geq 2$.

\begin{example}\label{33}
	We can construct three MUSEB$3$s in $\C^{6}\otimes \C^{6}$ from three MUMEBs in $\C^{3}\otimes \C^{3}$ and three MUBs in $\C^{2}$. Namely, $M_{3}(6,6)\geq 3$.   For three MUMEBs in $\C^{3}\otimes \C^{3}$, we construct the following three bases in $\mathcal{M}_{3\times 3}$ with $w=e^{\frac{2\pi i}{3}}$ and $i=\sqrt{-1}$.
\begin{widetext}
	\begin{align*}
	\mathcal{S}_{1}=\{&\frac{1}{3\sqrt{3}}\left(
	\begin{matrix}
	w+2  &w+2   & w+2\\
	2w^{2}+1 &w+2      &w^{2}+2w\\
	w^{2}+2w  &w+2      &2w^{2}+1
	\end{matrix}\right), \ \frac{1}{3\sqrt{3}}\left(
	\begin{matrix}
	w+2  &w^{2}+2w  &2w^{2}+1\\
	2w^{2}+1 &w^{2}+2w &w+2\\
	w^{2}+2w  &w^{2}+2w &w^{2}+2w
	\end{matrix}
	\right),
	\frac{1}{3\sqrt{3}}\left(
	\begin{matrix}
	w+2       &2w^{2}+1   &w^{2}+2w\\
	2w^{2}+1  &2w^{2}+1   &2w^{2}+1 \\
	w^{2}+2w  &2w^{2}+1   &w+2
	\end{matrix}
	\right),\\
	&\frac{1}{3\sqrt{3}}\left(
	\begin{matrix}
	w+2     & w+2        &w+2\\
	w+2     & w^{2}+2w   &2w^{2}+1\\
	w+2     & 2w^{2}+1   &w^{2}+2w
	\end{matrix}\right),
	\frac{1}{3\sqrt{3}}\left(
	\begin{matrix}
	w+2   &w^{2}+2w  &2w^{2}+1\\
	w+2   &2w^{2}+1  &w^{2}+2w\\
	w+2   &w+2        &w+2
	\end{matrix}
	\right), \ \frac{1}{3\sqrt{3}}\left(
	\begin{matrix}
	w+2  &2w^{2}+1   &w^{2}+2w\\
	w+2  &w+2        &w+2 \\
	w+2  &w^{2}+2w   &2w^{2}+1
	\end{matrix}
	\right),\\
	&\frac{1}{3\sqrt{3}}\left(
	\begin{matrix}
	w+2        &w+2        &w+2\\
	w^{2}+2w   &2w^{2}+1   &w+2\\
	2w^{2}+1   &w^{2}+2w   &w+2
	\end{matrix}\right), \ \frac{1}{3\sqrt{3}}\left(
	\begin{matrix}
	w+2        &w^{2}+2w    &2w^{2}+1 \\
	w^{2}+2w   &w+2         &2w^{2}+1\\
	2w^{2}+1   &2w^{2}+1   &2w^{2}+1
	\end{matrix}
	\right),
	\frac{1}{3\sqrt{3}}\left(
	\begin{matrix}
	w+2       &2w^{2}+1  &w^{2}+2w\\
	w^{2}+2w  &w^{2}+2w  &w^{2}+2w\\
	2w^{2}+1  &w+2       &w^{2}+2w
	\end{matrix}
	\right)\},\\
	\mathcal{S}_{2}=\{&\frac{1}{3\sqrt{3}}\left(
	\begin{matrix}
	0  & 3w^{2} & 0\\
	3w & 0      & 0\\
	0  & 0      & 3
	\end{matrix}\right), \ \frac{1}{3\sqrt{3}}\left(
	\begin{matrix}
	0  & 3 & 0\\
	3w & 0 & 0\\
	0  & 0 & 3w^{2}
	\end{matrix}
	\right), \ \frac{1}{3\sqrt{3}}\left(
	\begin{matrix}
	0  & 3w  & 0\\
	3w & 0   & 0\\
	0  & 0   & 3w
	\end{matrix}
	\right),
	\frac{1}{3\sqrt{3}}\left(
	\begin{matrix}
	3w^{2}  & 0 & 0\\
	0       & 0 & 3w\\
	0       & 3 & 0
	\end{matrix}\right), \ \frac{1}{3\sqrt{3}}\left(
	\begin{matrix}
	3w^{2}& 0   & 0\\
	0     & 0   & 3\\
	0     &3w   & 0
	\end{matrix}
	\right),\\
	&\frac{1}{3\sqrt{3}}\left(
	\begin{matrix}
	3w^{2}  & 0   & 0\\
	0       & 0   & 3w^{2}\\
	0       & 3w^{2}  & 0
	\end{matrix}
	\right),
	\frac{1}{3\sqrt{3}}\left(
	\begin{matrix}
	0       & 0    & 3w^{2}\\
	0       & 3w   & 0\\
	3       & 0 & 0
	\end{matrix}\right), \ \frac{1}{3\sqrt{3}}\left(
	\begin{matrix}
	0     &0    &3w\\
	0     &3w^{2}   & 0\\
	3     &0   & 0
	\end{matrix}
	\right), \ \frac{1}{3\sqrt{3}}\left(
	\begin{matrix}
	0       & 0   &3\\
	0       &3   &0\\
	3       &0  & 0
	\end{matrix}
	\right)\},\\
	\mathcal{S}_{3}=\{&\frac{1}{\sqrt{3}}\left(
	\begin{matrix}
	1  & 0  & 0\\
	0  & 1  & 0\\
	0  & 0  & 1
	\end{matrix}\right), \ \frac{1}{\sqrt{3}}\left(
	\begin{matrix}
	1  & 0 & 0\\
	0  & w & 0\\
	0  & 0 & w^{2}
	\end{matrix}
	\right), \ \frac{1}{\sqrt{3}}\left(
	\begin{matrix}
	1  & 0      & 0\\
	0  & w^{2}  & 0\\
	0  & 0      & w
	\end{matrix}
	\right),
	\frac{1}{\sqrt{3}}\left(
	\begin{matrix}
	0   & 0 & 1\\
	1   & 0 & 0\\
	0   & 1 & 0
	\end{matrix}\right), \ \frac{1}{\sqrt{3}}\left(
	\begin{matrix}
	0     & 0   & w^{2}\\
	1     & 0   & 0\\
	0     & w   & 0
	\end{matrix}
	\right),
	\frac{1}{\sqrt{3}}\left(
	\begin{matrix}
	0       & 0      & w\\
	1       & 0      &0\\
	0       & w^{2}  & 0
	\end{matrix}
	\right),\\
	&\frac{1}{\sqrt{3}}\left(
	\begin{matrix}
	0       & 1   & 0\\
	0       & 0   & 1\\
	1       & 0   & 0
	\end{matrix}\right), \ \frac{1}{\sqrt{3}}\left(
	\begin{matrix}
	0     &w    &0\\
	0     &0   & w^{2}\\
	1     &0   & 0
	\end{matrix}
	\right), \ \frac{1}{\sqrt{3}}\left(
	\begin{matrix}
	0       &w^{2}   &0\\
	0       &0       &w\\
	1       &0      & 0
	\end{matrix}
	\right)\}.
	\end{align*}
For three MUBs in $\C^{2}$, choose
	$$\mathcal{T}_{1}=\{(0,1),(1,0)\}, \ \mathcal{T}_{2}=\{\frac{1}{\sqrt{2}}(1,1),\frac{1}{\sqrt{2}}(1,-1)\}, \
	\mathcal{T}_{3}=\{\frac{1}{\sqrt{2}}(1,i),\frac{1}{\sqrt{2}}(1,-i)\}.$$
\end{widetext}
Then $\{\mathcal{S}_{1}\otimes \mathcal{T}_{1}, \mathcal{S}_{2}\otimes \mathcal{T}_{2},\mathcal{S}_{3}\otimes \mathcal{T}_{3}\}$ is a set of three MUMEBs in $\C^{3}\otimes \C^{6}$ by Eq.~(\ref{mdq}), and  $\{\mathcal{T}_{1}^{\mathrm{T}}\otimes \mathcal{S}_{1}\otimes \mathcal{T}_{1}, \mathcal{T}_{2}^{\mathrm{T}}\otimes\mathcal{S}_{2}\otimes \mathcal{T}_{2},\mathcal{T}_{3}^{\mathrm{T}}\otimes\mathcal{S}_{3}\otimes \mathcal{T}_{3}\}$ is a set of three MUSEB$3$s in
$\C^{6}\otimes \C^{6}$ by Eq.~(\ref{mpd}).	
\end{example}

\section{Lower Bounds for MUMEBs}\label{Mbs}


Through out this section, we always assume that $d=2^{t}p_{1}^{a_{1}}p_{2}^{a_{2}}\cdots p_{s}^{a_{s}}$, where $p_{1},p_{2},\ldots p_{s}$ are distinct odd primes such that $3\leq p_{1}^{a_{1}}\leq p_{2}^{a_{2}}\leq \cdots \leq p_{s}^{a_{s}}$. For convenience, we define $M_{1}(1,1):=\infty$.

\begin{theorem}\label{dddlower}
	For any positive integer $d\geq 2$, we have $M_{d}(d,d)\geq 3$. In particular,
	\begin{enumerate}[(i)]
		\item $M_{d}(d,d)\geq \frac{p_{1}^{2a_{1}}-1}{2}$ when $t=0$;
		\item $M_{d}(d,d)\geq 3$ when $t=1$;
		\item $M_{d}(d,d)\geq \min\{2(2^t-1),\frac{p_{1}^{2a_{1}}-1}{2}\}$ when $t\geq 2$.
	\end{enumerate}
\end{theorem}
\begin{proof} By Eq.~(\ref{meb}), we have
\begin{equation*}
  M_{d}(d,d)\geq
  \min\{M_{2^t}(2^t,2^t),M_{p_{1}^{a_{1}}}(p_{1}^{a_{1}},p_{1}^{a_{1}}),\ldots,M_{p_{s}^{a_{s}}}(p_{s}^{a_{s}},p_{s}^{a_{s}})\}.
\end{equation*}
Since $M_{p_{i}^{a_{i}}}(p_{i}^{a_{i}},p_{i}^{a_{i}})\geq\frac{p_{i}^{2a_{i}}-1}{2}$  by \cite{Xu1},  $M_{d}(d,d)\geq\min\{M_{2^t}(2^t,2^t), \frac{p_{1}^{2a_{1}}-1}{2}\}$. Then the three cases follow from the fact that $M_1(1,1)=\infty$, $M_{2}(2,2)=3$ and $M_{2^t}(2^t,2^t)\geq 2(2^t-1)$ by \cite{Xu}.
\end{proof}

Theorem~\ref{dddlower} gives a general lower bound $M_{d}(d,d)\geq 3$ for any $d\geq 2$, which is similar to the case of MUBs $N(q)\geq 3$ for any $q\geq 2$.

By Eq.~(\ref{mdq}) and Eq.~(\ref{mpd}), we have the following corollary.
\begin{cor}
\label{ddqdl}  For all integers $d, p, q\geq 2$, we have
	\begin{equation}\label{dkd1}
	M_{d}(d,qd)\geq \min\{N(q),M_{d}(d,d)\}\geq 3,
	\end{equation}
	\begin{equation}
	M_{d}(pd,qd)\geq \min\{N(p),M_{d}(d,qd)\}\geq 3. 
	\end{equation}
\end{cor}

\begin{remark}
	In \cite{Cheng}, it was shown that $M_{d}(d,qd)\geq \min\{(p_{1}')^{a_{1}'}+1,M_{d}(d,d)\}$ when $d$ is an \emph{odd} number.
	Eq. (\ref{dkd1}) extends this result for any number $d\geq 2$.
\end{remark}

We have the following corollary by applying Eq.~(\ref{meb})  with Theorem~\ref{C23} and Eq.~ (\ref{dkd1}).

\begin{cor}\label{21k}
	There are at least two MUMEBs in $\C^{2d}\otimes \C^{3qd}$ for any positive integers $q$ and $d$. Namely,
	\begin{equation}\label{002}
	M_{2d}(2d,3qd)\geq 2.
	\end{equation}
	There are at least two MUSEB$2$s in $\C^{3}\otimes \C^{2k}$ for any $k\geq 1$. Namely,
	\begin{equation}\label{112}
	M_{2}(3,2k)\geq 2.
	\end{equation}
\end{cor}

\begin{remark}
	Corollary~\ref{21k} shows that two MUMEBs exist  in $\C^{d}\otimes \C^{d'}$ for infinitely many parameters $d,d'$ satisfying $d\nmid d'$. The result of Eq.~(\ref{112}) is better than that in \cite{Han}, which only showed that $M_{2}(3,4k)\geq 2$  where $k=2^{l}$.
\end{remark}

\begin{example}
	There are two MUMEBs in $\C^{6}\otimes \C^{9}$, $\mathcal{R}_{1}\otimes \mathcal{S}_{1},\mathcal{R}_{2}\otimes \mathcal{S}_{2}$. Namely, $M_{6}(6,9)\geq 2$. Here, $\mathcal{R}_{1}, \mathcal{R}_{2}$ are from Sec.~\ref{sec:first}, and $\mathcal{S}_{1}, \mathcal{S}_{2}$ are constructed in Example~\ref{33}.
%
%
\end{example}

\section{Conclusion}
\label{conpen}

In this paper we  studied constructions of MUSEB$k$s (MUMEBs) and provided several better lower bounds on the maximum size of MUSEB$k$s. See Table~\ref{all} for a conclusion of all known results. By similar arguments as in \cite{Shaari}, one can consider the issue of distinguishability of SEB$k$s selected from a set of MUSEB$k$s and the use of MUSEB$k$s in a QKD setup. We can also consider the mutually unbiased measurements consisting of MUSEB$k$s, which plays a special role in the problem of state determination \cite{Woot}.

For the existence of  MUSEB$k$s, there are still many open questions. For example, which $d$ can achieve the upper bound $M_{d}(d,d)\leq d^{2}-1$ besides $2,3,5,7,11$? Can we improve the lower bound of $M_{d}(d,qd)\geq 3$?  Are there three MUMEBs in $\C^{2}\otimes \C^{3}$,  while there are three MUBs in $\C^{6}$. Are there MUSEB$k$s in $\C^{d}\otimes \C^{d'}$
when $k\nmid d d'$? Are there MUMEB  in $\C^{d}\otimes \C^{d'}$ for any $d\nmid d'$? The minimum unsolved case is $\C^{2}\otimes \C^{5}$.

\begin{acknowledgements}	
FS and XZ were supported by NSFC under Grant No. 11771419,  the Fundamental Research Funds for the Central Universities,	and Anhui Initiative in Quantum Information Technologies under Grant No. AHY150200. LC was supported by the  NNSF of China (Grant No. 11871089), and the Fundamental Research Funds for the Central Universities (Grant Nos. KG12080401 and ZG216S1902).

\end{acknowledgements}


\bibliographystyle{unsrt}

\bibliography{MUSEBk}

\end{document}